\documentclass[11pt]{article}

\usepackage{color}
\usepackage[margin=1in]{geometry}
\usepackage[T1]{fontenc}
\usepackage{lmodern}
\usepackage{amsmath}
\usepackage{amssymb}
\usepackage{amsthm}
\usepackage{hyperref}
\usepackage{graphicx}

\newtheorem{theorem}{Theorem}[section]
\newtheorem{lemma}[theorem]{Lemma}

\newtheorem{definition}[theorem]{Definition}

\newcommand{\eps}{\varepsilon}
\renewcommand{\epsilon}{\varepsilon}
\newcommand{\Rbb}{\mathbb{R}}
\newcommand{\R}{\mathbb{R}}

\newcommand{\A}{{\cal A}}
\newcommand{\G}{{\cal G}}
\renewcommand{\d}{{\mathsf{D}}}

\newcommand{\todo}[2]{}

\begin{document}
\title{Approximate Nearest Neighbor Search in High Dimensions}
\author{Alexandr Andoni \and Piotr Indyk \and Ilya Razenshteyn}
\date{}
\maketitle
\begin{abstract}
The nearest neighbor  problem is defined as follows: Given a set $P$ of $n$ points in some metric space $(X,\d)$, build a data structure that, given any  point $q$, returns a point in $P$ that is closest to $q$ (its ``nearest neighbor'' in $P$).  The data structure  stores additional information about the set $P$, which is then used to find the nearest neighbor without computing all distances between~$q$ and $P$. The problem has a wide range of applications in machine learning, computer vision, databases and other fields.

To reduce the time needed to find nearest neighbors  and the amount of memory used by the data structure, one can formulate the {\em approximate} nearest neighbor problem, where the the goal is to return any point $p' \in P$ such that the distance from $q$ to $p'$ is at most $c \cdot \min_{p \in P} \d(q,p)$,  for some  $c\ge 1$. Over the last two decades, many efficient solutions to this problem were developed. In this article we survey these developments, as well as their connections to  questions in geometric functional analysis and combinatorial geometry.

\end{abstract}
\section{Introduction}

The {\em nearest neighbor}  problem is defined as follows: Given a set $P$ of $n$ points in a metric space defined over a set $X$ with distance function $\d$, build a data structure\footnote{See Section \ref{ss:model} for a discussion about the computational model.} that, given any ``query'' point $q \in X$, returns its ``nearest neighbor''  $\arg \min_{p \in P} \d(q,p)$. A particularly interesting and well-studied case is that of nearest neighbor in geometric spaces, where $X = \R^d$ and the metric~$\d$ is induced by some norm. The problem has a wide range of applications in machine learning, computer vision, databases and other fields, see~\cite{shakhnarovich2006nearest, andoni2008near} for an overview.

 A simple solution to this problem would  store the set $P$ in memory, and then,  given $q$,  compute all distances $\d(q,p)$ for $p \in P$ and select the point  $p$ with the minimum distance. Its disadvantage  is the computational cost: computing all $n$ distances  requires  at least $n$ operations. Since in many applications $n$ can be as large as $10^9$ (see e.g.,~\cite{sundaram2013streaming}),  it was necessary to develop faster methods that find the nearest neighbors without explicitly computing all distances from $q$.  Those methods compute and store additional information about the set $P$, which is then used to find nearest neighbors more efficiently.  To illustrate this idea, consider another simple solution for the case where $X=\{0,1\}^d$. In this case, one could precompute and store in memory the answers to all $2^d$ queries $q \in X$;  given $q$, one could then return its nearest neighbor by performing only a single memory lookup. Unfortunately, this approach requires memory of size $2^d$, which again is inefficient ($d$ is at least $10^3$ or higher in many applications).

The two aforementioned solutions can be viewed as extreme points in a
tradeoff between the time
to answer a query (``query time'') and the amount of memory used (``space'')\footnote{There are other important
  data structure parameters, such as the time needed to construct
  it. For the sake of simplicity, we will mostly focus on query time
  and space.}.   The study of this tradeoff dates back to the work of
Minsky and Papert (\cite{minsky1969perceptrons}, p. 222), and has
become one of the key topics in the field of {\em computational
  geometry}~\cite{preparata1985introduction}. During the 1970s and
1980s many efficient solutions have been discovered for the case when
$(X,\d)=(\R^d, \ell_2)$ and $d$ is a constant independent of $n$.  For
example, for $d=2$, one can construct a data structure using $O(n)$
space with $O(\log n)$ query time~\cite{lipton1980applications}. Unfortunately, as the dimension
$d$ increases, those data structures become less and less
efficient. Specifically, it is known how construct data structures with
$O(d^{O(1)} \log n)$ query time, but using $n^{O(d)}$ space
(\cite{meiser1993point}, building
on~\cite{clarkson1988randomized}).\footnote{This exponential
  dependence on the dimension is due to the fact that those data
  structures compute and store the {\em Voronoi decomposition} of $P$,
  i.e., the decomposition of $\R^d$ into cells such that all points in
  each cell have the same nearest neighbor in $P$.  The combinatorial
  complexity of this decomposition could be as large as
  $n^{\Omega(d)}$ \cite{caratheodory1911variabilitatsbereich}.  } Furthermore, there is
evidence
   that  data structures with query times  of the form $n^{1-\alpha} d^{O(1)}$ for some constant $\alpha>0$ might be difficult to construct efficiently.\footnote{If such a data structure could be constructed in
  polynomial time $n^{O(1)} d^{O(1)}$, then the Strong Exponential
  Time Hypothesis~\cite{virgi} would be
  false. This fact essentially follows from~\cite{williams2005new}, see the discussion after Theorem 1 in~\cite{ahle2016complexity}.}

The search for efficient solutions to the nearest neighbor problem has led to the question whether better space/query time bounds could be obtained if the data structure was allowed to report {\em approximate} answers. In the {\em $c$-approximate nearest neighbor} problem, the data structure can report any point $p' \in P$ within distance $c \cdot \min_{p \in P} \d(q,p)$ from $q$; the parameter $c \ge 1$ is called ``approximation factor''.  The work of Arya, Mount and Bern~\cite{arya1993approximate,bern1993approximate} showed that allowing $c>1$  indeed leads to better data structures, although their solutions still retained exponential dependencies on $d$ in the query time or space bounds~\cite{arya1993approximate}  or required that the approximation factor $c$ is polynomial in the dimension $d$~\cite{bern1993approximate}. These bounds have been substantially improved over the next few years, see e.g., \cite{clarkson1994algorithm,chan1998approximate,arya1998optimal,kleinberg1997two} and the references therein.

In this article we survey the ``second wave'' of approximate nearest neighbor data structures, whose query time and space bounds are polynomial in the dimension $d$. \footnote{Due to the lack of space, we will not cover several important related topics, such as data structures for point-sets with low intrinsic dimension~\cite{clarkson2006nearest}, approximate furthest neighbor, approximate nearest line search~\cite{mahabadi2014approximate} and other variants of the problem.} At a high level, these data structures are obtained in two steps. In the first step,  the approximate {\em nearest} neighbor problem is reduced to its ``decision version'', termed approximate {\em near} neighbor (see e.g.~\cite{har2012approximate}). The second step involves constructing a data structure for the latter problem. In this survey we focus mostly on the second step.

The approximate near neighbor problem is  parameterized by an approximation factor $c>1$ as well as a ``scale parameter'' $r>0$, and defined as follows.

\begin{definition}[$(c,r)$-Approximate Near Neighbor]
Given a set $P$ of $n$ points in a metric space $(X,\d)$, build a data structure ${\cal S}$ that, given any
query point $q \in X$ such that the metric ball $B_{\d}(q,r)=\{p \in X:
\d(p,q)\le r\}$ contains a point in $P$,  ${\cal S}$ returns any point in
$B_\d(q,cr) \cap P$.
\end{definition}

Note that the definition does not specify the behavior of the data structure if the ball $B_\d(q,r)$ does not contain any point in $P$. We omit the index $\d$ when it is clear from the context.

The above definition applies to algorithms that are {\em deterministic}, i.e., do not use random bits. However, most of the approximate near neighbor algorithms in the literature are {\em randomized}, i.e., generate and use random bits  while constructing the data structure. In this case, the data structure~${\cal S}$ is a random variable, sampled from some distribution. This leads to the following generalization.

\begin{definition}[$(c,r,\delta)$-Approximate Near Neighbor]
\label{d:crdann}
Given a set $P$ of $n$ points in a metric space $(X,\d)$, build a data structure ${\cal S}$ that, given any
query point $q \in X$ such that  $B(q,r) \cap P \neq \emptyset$,
\[  \Pr_{\cal S}[ {\cal S}  \mbox{\ returns any point in\ } B(q,cr) \cap P] \ge 1-\delta \]
\end{definition}

The probability of failure $\delta$ of the data structure can be reduced by independently repeating the process several  times, i.e., creating  several data structures.  Therefore,  in the rest of the survey we will set $\delta$ to an arbitrary constant, say, $1/3$. We will use $(c,r)$-ANN to denote $(c,r,1/3)$-Approximate Near Neighbor.

\subsection{Computational model}
\label{ss:model}
For the purpose of this survey, a data structure of size $M$ is an array $A[1 \ldots M]$ of numbers (``the memory''), together with an associated algorithm that, given a point $q$, returns a point in $P$ as specified by the problem. The entries $A[i]$ of $A$ are called ``memory cells''. Due to lack of space, we will not formally define other details of the computational model, in particular what an algorithm  is, how to measure its running time, what is the range of the array elements $A[i]$, etc. There are several ways of defining these notions, and the material in this survey is relatively robust to the variations in the definitions.  We note, however, that one way to formalize these notions is to restrict all numbers, including point coordinates, memory entries, etc, to rational numbers of the form $a/b$, where  $a  \in \{-n^{O(1)} \ldots n^{O(1)}\}$ and $b=n^{O(1)}$,  and to define query time as the maximum number of memory cells accessed to answer any query $q$.

For an overview of these topics and formal definitions, the reader is referred to~\cite{miltersen1999cell}.  For
a~discussion specifically geared towards mathematical audience,  see \cite{fefferman2009fitting}.

\section{Data-independent approach}
\label{s:di}

The first approach to the approximate near neighbor problem has been
via data-independent data structures. These are data structures where
the memory cells accessed by the query algorithm do not depend on the
data set $P$, but only on $q$ and (for randomized data structures) the
random bits used to construct the data structure.  In this section, we
describe two methods for constructing such data structures, based on
{\em oblivious dimension-reduction}, and on {\em randomized space
  partitions}. These methods give ANN data structures for the $\ell_1$
and $\ell_2$ spaces in particular.

\subsection{ANN via dimension reduction}

As described in the introduction, there exist ANN data
structures with space and query time at most exponential in the
dimension $d$. Since exponential space/time bounds are unaffordable for large $d$, a natural approach is
to perform a {\em dimension reduction} beforehand, and then solve the
problem in the lower, reduced dimension. The main ingredient of such
an approach is a map $f:\R^d\to \R^k$ that preserves
distances up to a $c=1+\eps$ factor, where $k=O(\log n)$. Then a space bound
exponential in $k$ becomes polynomial in $n$.

Such dimension-reducing maps $f$ indeed exist for the $\ell_2$ norm if we allow
randomization, as first shown in the influential paper by Johnson and Lindenstrauss:

\begin{lemma}[\cite{johnson1984extensions}]
\label{lem:JL}
Fix dimension $d\ge 1$ and a ``target'' dimension $k<d$. Let $A$ be the
projection of $\R^d$ to  its $k$-dimensional subspace selected uniformly at random (with respect to the Haar measure), and define $f:\R^d\to \R^k$ as $f(x)=\tfrac{\sqrt{d}}{\sqrt{k}}Ax$. Then,
there is a universal constant $C>0$, such that for any $\eps\in
(0,1/2)$, and any $x,y\in \R^d$, we have that
$$
\Pr_A\left[\tfrac{\|f(x)-f(y)\|}{\|x-y\|}\in (1-\eps,
  1+\eps)\right]\ge 1-e^{-C\eps^2k}.
$$
\end{lemma}

We can now apply this lemma, with  $k=O\left(\tfrac{\log
  n}{\eps^2}\right)$, to a set of points $P$ to show that the map~$f$ has a
$(1+\eps)$ distortion on $P$, with probability at least $2/3$. Most
importantly the map $f$ is ``oblivious'', i.e., it does not depend on
$P$.

We now show how to use Lemma~\ref{lem:JL} to design a
$(1+O(\eps), r)$-ANN data structure with the following guarantees.
\begin{theorem}[\cite{indyk1998approximate, har2012approximate}]
\label{t:im}
Fix $\eps\in(0,1/2)$ and dimension $d\ge 1$. There is a $(1+O(\eps), r)$-ANN data
structure over $(\R^d, \ell_2)$ achieving $Q=O(d\cdot \tfrac{\log
  n}{\eps^2})$ query time, and $S=n^{O(\log(1/\eps)/\eps^2)}+O(d(n+k))$ space.
  The time needed to build the data structure is $O(S+ndk)$.
\end{theorem}
\begin{proof}[Proof sketch]
First, assume there is a $(1+\eps, r)$-ANN data structure $\A$ for the
$k$-dimensional $\ell_2$ space, achieving query time $Q(n,k)$ and
space bounded by $S(n,k)$. For $k=O(\tfrac{\log
  n}{\eps^2})$, we consider the map $f$ from Lemma~\ref{lem:JL}. For the
dataset $P$, we compute $f(P)$ and preprocess this set using $\A$
(with the scale parameter $r(1+\eps)$). Then, for a query point $q\in
\R^d$, we query the data structure $\A$ on $f(q)$. This algorithm
works for a fixed dataset $P$ and query $q$ with 5/6 probability,
by applying Lemma \ref{lem:JL} to the points in the set $P\cup
\{q\}$.
  The map $f$ preserves all distances between $P$ and
$q$ up to a factor of  $1+\eps$.

We now construct $\A$ with space $S(n,k)=n\cdot
(1/\eps)^{O(k)}$ and $Q(n,k)=O(k)$, which yields the stated bound for
$k=O(\tfrac{\log n}{\eps^2})$. Given the scale parameter $r$, we
discretize the space $\R^k$ into cubes of sidelength $\eps
r/\sqrt{k}$, and consider the set $S$ of cubes that intersect any ball
$B(p',r)$ where $p'\in f(P)$. Using standard estimates on the volume
of $\ell_2$ balls, one can prove that $|S|\le n\cdot
(1/\eps)^{O(k)}$. The data structure then stores the set $S$ in a
dictionary data structure.\footnote{In the dictionary problem, we are
  given a set $S$ of elements from a discrete universe $U$, and we
  need to answer queries of the form ``given $x$, is $x\in S$?''. This
  is a classic data structure problem and has many solutions. One
  concrete solution is via a hashing~\cite{CLRS}, which achieves space of
  $O(|S|)$ words, each of $O(\log |U|)$ bits, and query time
  of $O(1)$ in expectation.} For a query $f(q)$, we just compute the
cube that contains $f(q)$, and check whether it is contained in set $S$ using the
dictionary data structure. We note that there is an additional
$1+\eps$ factor loss from discretization since the diameter of a cube
is $\eps r$.
\end{proof}

A similar approach was  introduced
\cite{kushilevitz2000efficient} in the context of the Hamming space $\{0,1\}^d$.
An important difference is that that there is no analog
 of Lemma~\ref{lem:JL} for the Hamming space
\cite{brinkman2005impossibility}.\footnote{In fact, it has been shown that spaces for which analogs of Lemma~\ref{lem:JL} hold are ``almost'' Hilbert spaces~\cite{johnson2009johnson}.}
Therefore,
\cite{kushilevitz2000efficient} introduce a weaker notion of randomized
dimension reduction, which works only for a {\em fixed scale} $r$.

\begin{lemma}[\cite{kushilevitz2000efficient}]
\label{lem:KORdr}
Fix the error parameter $\eps\in(0,1/2)$, dimension $d\ge 1$, and scale $r\in
[1,d]$. For any $k\ge 1$, there exists a randomized map
$f:\{0,1\}^d\to \{0,1\}^k$ and an absolute constant $C> 0$,
satisfying the following for any fixed $x,y\in \{0,1\}^d$:
\begin{itemize}
\item
if $\|x-y\|_1\le r$, then $\Pr_f[\|f(x)-f(y)\|_1\le k/2]\ge
1-e^{-C\eps^2k}$;
\item
if $\|x-y\|_1\ge (1+\eps)r$, then $\Pr_f[\|f(x)-f(y)\|_1>(1+\epsilon/2)\cdot k/2]\ge
1-e^{-C\eps^2k}$.
\end{itemize}
\end{lemma}

The map $f$ can be constructed via a random projection over $GF(2)$.
That is,  take $f(x)=Ax$, where $A$ is a $k\times d$ matrix for $k=O(\log(n)/\epsilon^2)$, with each
entry being $1$ with some fixed probability~$p$, and zero
otherwise. The probability $p$ depends solely on $r$. The rest of the
algorithm proceeds as before, with the exception that the ``base''
data structure $\A$ is particularly simple: just store the answer for
any dimension-reduced query point $f(q)\in \{0,1\}^k$. Since there
are only $2^k=n^{O(1/\eps^2)}$ such possible queries, and computing $f(q)$ takes $O(dk)$ time,  we get the
following result.

\begin{theorem}[\cite{kushilevitz2000efficient}]
\label{t:KOR}
Fix $\eps\in(0,1/2)$ and dimension $d\ge 1$. There is a $(1+O(\eps), r)$-ANN data
structure over $(\{0,1\}^d, \ell_1)$ using
$n^{O(1/\eps^2)}  + O(d(n+k))$ space and $O(d\cdot \tfrac{\log
  n}{\eps^2})$ query time.
\end{theorem}

As a final remark, we note we cannot obtain improved space bounds by
improving the dimension reduction lemmas \ref{lem:JL} and
\ref{lem:KORdr}. Indeed the above lemmas are tight as proven in
\cite{jayram2013optimal}. There was however work on improving the {\em
  run-time complexity} for computing a dimension reduction map,
improving over the na\"ive bound of $O(dk)$; see \cite{AC-fastJL, DKS10,
  AL13, kw11-rip, npw14-rip, kn14-sparseJL}.

\subsection{ANN via space partitions: Locality-Sensitive Hashing}
\label{ss:lsh}
While dimension reduction yields ANN  data structure with polynomial space, this is
not enough in applications, where one desires space as close as
possible to linear in $n$. This consideration led to the following,
alternative approach, which yields smaller space bounds, albeit at the
cost of increasing the query time to something of the form $n^\rho$
where $\rho\in(0,1)$.

The new approach is based on randomized space partitions, and
specifically on Locality-Sensitive Hashing, introduced in
\cite{indyk1998approximate}.

\begin{definition}[Locality-Sensitive Hashing (LSH)]
Fix a metric space $(X, \d)$, scale $r>0$,  approximation
$c>1$ and a set $U$. Then a distribution $\cal H$ over maps $h:X\to U$ is called
$(r,cr, p_1,p_2)$-sensitive if the following holds for any $x,y\in X$:
\begin{itemize}
\item
if $\d(x,y)\le r$, then $\Pr_{h}[h(x)=h(y)]\ge p_1$;
\item
if $\d(x,y)> cr$, then $\Pr_{h}[h(x)=h(y)]\le p_2$.
\end{itemize}
The distribution $\cal H$ is called an LSH family, and has {\em
  quality} $\rho=\rho({\cal H})=\tfrac{\log 1/p_1}{\log 1/p_2}$.
\end{definition}
In what follows we require an LSH family to have $p_1>p_2$, which implies $\rho<1$.
Note that LSH mappings are also oblivious: the distribution $\cal H$
does not depend on the point-set $P$ or the query $q$.

Using LSH, \cite{indyk1998approximate} show how to obtain
the following ANN data structure.

\begin{theorem}[\cite{indyk1998approximate}]
\label{thm:LSHds}
Fix a metric ${\cal M}=(X, \d)$, a scale $r>0$, and approximation
factor $c>1$. Suppose the metric admits a $(r,cr,p_1,p_2)$-sensitive
LSH family ${\cal H}$, where the map $h(\cdot)$ can be stored in
$\sigma$ space, and, for given $x$, can be computed in $\tau$ time;
similarly, assume that computing distance $\d(x,y)$ takes $O(\tau)$
time. Let $\rho=\rho({\cal H})=\tfrac{\log 1/p_1}{\log 1/p_2}$. Then
there exists a $(c,r)$-ANN data structure over ${\cal M}$ achieving
query time $Q=O(n^\rho\cdot \tau\tfrac{\log_{1/p_2} n}{p_1})$ and
space $S=O(n^{1+\rho}\cdot \tfrac{1}{p_1}+n^\rho\tfrac{1}{p_1}\cdot
\sigma\cdot \log_{1/p_2} n)$ (in addition to storing the original
dataset $P$).  The time needed to build this data structure is $O(S\cdot \tau)$.
\end{theorem}

While we describe some concrete LSH families later on, for now, one
can think of the parameters~$\tau, \sigma$ as being proportional to the
dimension of the space (although this is not always the case).

The overall idea of the algorithm is to use an LSH family as a
pre-filter for the dataset $P$. In particular, for a random partition
$h$ from the family ${\cal H}$, the query point $q$ will likely
collide with its near neighbor (with probability at least $p_1$), but
with few points at a distance $\ge cr$, in expectation at most $p_2\cdot
n$ of them. Below we show how an extension of this idea yields
Theorem~\ref{thm:LSHds}.

\begin{proof}[Proof sketch]
Given an LSH family $\cal H$, we can build a new, derived LSH family via
a certain tensoring operation.  In particular, for an integer $k\ge 1$,
consider a new distribution $\G_k$ over maps $g:X\to U$, where
$g(\cdot)$ is obtained by picking $k$  i.i.d. functions $h_1,\ldots h_k$ chosen from $\cal{H}$ and setting $g(x)=(h_1(x),h_2(x),\ldots h_k(x))$ Then, if $\cal H$ is $(r,cr,p_1,p_2)$-sensitive, $\G_k$ is
$(r,cr,p_1^k,p_2^k)$-sensitive. Note that the parameter $\rho$ of the
hash family does not change, i.e.,  $\rho({\G_k})=\rho({\cal H})$.

The entire ANN data structure is now composed of $L$ dictionary data
structures (e.g., hash tables discussed in the previous section), where $L,k\ge 1$ are
parameters to fix later. The $i$-th hash table is constructed as
follows. Pick a map $g_i$ uniformly at random from $\G_k$, and store the set $g_i(P)$ in the dictionary structure.
At the query time,  we iterate over $i=1 \ldots L$.
For a~given~$i$, we compute $g_i(q)$, and use the dictionary structure to obtain the set of
``candidate'' points $Q_i = \{ p \in P: g_i(p)=g_i(q) \}$.
For each candidate point we compute the distance from $q$ to that point.
The process is stopped when all $Q_i$'s are processed, or when a point within distance $cr$ to $q$ is found, whichever happens first.

To analyze the success probability, we note that the dictionary
structure $i$ succeeds if $p^*\in Q_i$, where $p^*$ is the assumed
point at distance at most $r$ from $q$. This happens with probability
at least~$p_1^k$. Thus, we can take $L=O(1/p_1^k)$ such dictionary
structures, and thus guarantee success with a constant probability.

The expected query time is $O(L(k\tau+Ln\cdot p_2^k\cdot \tau))$, which includes
both the computation of the maps $g_1(q),\ldots g_L(q)$ and the of
distances to the candidates in sets $Q_1,\ldots Q_L$. We can now derive the
value of $k$ that minimizes the above, obtaining $k=\lceil\log_{1/p_2}
n\rceil\le \log_{1/p_2} n+1$, and hence $L=O(n^\rho/p_1)$.

Finally, note that the space usage is $O(Ln)$ for the
dictionary structures, plus $O(Lk\sigma)$ for the description of the
maps.
\end{proof}

\subsection{Space partitions: LSH constructions}

Theorem~\ref{thm:LSHds} assumes the existence of an LSH family $\cal
H$ with a parameter $\rho<1$. In what follows we show a few examples
of such families.

\begin{enumerate}
\item {\em Hamming space $\{0,1\}^d$, with $\rho=1/c$.}
  The
  distribution $\cal H$ is simply a projection on a random coordinate
  $i$: ${\cal H}=\{h_i: h_i(x)=x_i, i=1,\ldots d\}$. This family is
  $(r,cr, 1-r/d,1-cr/d)$-sensitive, and hence $\rho\le 1/c$  \cite{indyk1998approximate}.

This LSH scheme is near-optimal for the Hamming space, as described in
Section \ref{sec:lshLB}. We also note that, since $\ell_2$ embeds
isometrically into $\ell_1$ (see Section~\ref{s:extensions}), this result extends to $\ell_2$
as well.

\item {\em Euclidean space $(\R^d, \ell_2)$, with $\rho<1/c$.}  In
  \cite{datar2004locality}, the authors introduced an LSH family which slightly improves over the above construction. It is
  based on random projections in~$\ell_2$. In particular, define a
  random map $h(x)$ as $h(x)=\lfloor\tfrac{\langle
    x,g\rangle}{wr}+b\rfloor$, where $g$ is a random $d$-dimensional
  Gaussian vector, $b\in [0,1]$, and $w>0$ is a fixed parameter. It
  can be shown that, for any fixed $c>1$, there exists $w>0$ such that
  $\rho<1/c$.

\item {\em Euclidean space $(\R^d, \ell_2)$, with $\rho\to 1/c^2$.}
  In \cite{andoni2006near}, the authors showed an LSH family with a~much better $\rho$, which later turned out to be optimal (see
  Section~\ref{sec:lshLB}). At its core, the main idea is to partition
  the space into Euclidean balls.\footnote{In contrast, the above LSH family can
  be seen as partitioning the space into cubes: when considering the
  $k$-tensored family $\G={\cal H}^k$, the resulting map $g\in \G$ is
  equivalent to performing a random dimension reduction (by
  multiplying by a random $k\times d$ Gaussian matrix), followed by
  discretization of the space into cubes.}
  It proceeds in two steps: 1)
  perform a~random dimension reduction $A$ to dimension $t$ (a
  parameter), and 2) partition $\R^t$ into balls. Since it is
  impossible to partition the space $\R^t$ into balls
  precisely\footnote{This is also termed {\em tessellation} of the space.}
  when $t\ge 2$, instead one performs ``ball carving''. The basic idea
  is to consider a sequence of randomly-centered balls $B_1,B_2,
  \ldots $, each of radius $wr$ for some parameter $w>1$, and define
  the map $h(x)$, for $x\in \R^d$, to be the index $i$ of the first
  ball $B_i$ containing the point $Ax$. Since we want to cover an
  infinite space with balls of finite volume, the above procedure needs to be modified slightly to terminate in finite time.
  The modified procedure runs in time $T=t^{O(t)}$.

  Overall, optimizing for $w$, one can obtain $\rho=1/c^2+\tfrac{O(\log
    t)}{\sqrt{t}}$ which tends to $1/c^2$ as $t \to \infty$. The time to hash is
  $\tau=O(Tt+dt)$, where $T$ depends exponentially on the parameter
  $t$, i.e., $T=t^{\Theta(t)}$. For the ANN data structure, the optimal choice is $t=O(\log
  n)^{2/3}$, resulting in $\rho=1/c^2+\tfrac{O(\log \log n)}{(\log
    n)^{1/3}}$.
\end{enumerate}

The $\ell_2$ LSH families can be extended to other $\ell_p$'s. For
$p<1$, \cite{datar2004locality} showed one can use method 2 as
described above, but using $p$-stable distributions instead of
Gaussians.  See Section \ref{s:extensions} for other extensions for $p>1$.

We remark that there is a number of other widely used LSH families, including
min-hash \cite{broder1997resemblance, broder1997syntactic} and
simhash \cite{charikar2002similarity}, which apply to different notions of similarity between points.
See~\cite{andoni2008near} for an overview.

\subsection{Space partitions: impossibility results}
\label{sec:lshLB}

It is natural to explore the limits of LSH families and ask what is
the best $\rho$ one can obtain for a given metric space as a function
of the approximation $c>1$. In \cite{motwani2007lower, o2014optimal},
it was proven that the LSH families
\cite{indyk1998approximate,andoni2006near} from the previous section
are near-optimal: for the Hamming space, we must have $\rho\ge
1/c-o(1)$, and for the Euclidean space, $\rho\ge 1/c^2-o(1)$.  Below
is the formal statement from \cite{o2014optimal}.

\begin{theorem}
Fix dimension $d\ge 1$ and approximation $c\ge 1$. Let $\cal H$ be a
$(r,cr,p_1,p_2)$-sensitive LSH family over the Hamming space, and
suppose $p_2\ge 2^{-o(d)}$. Then $\rho\ge 1/c-o_d(1)$.
\end{theorem}

Note that the above theorem also immediately implies $\rho\ge
1/c^2-o(1)$ for the Euclidean space, by noting that  $\|x-y\|_1=\|x-y\|_2^2$ for binary vectors $x$ and $y$.

Finally, we remark that some condition on $p_2$ is necessary, as there
exists an LSH family with $p_2=0$, $p_1=2^{-O(d)}$ and hence
$\rho=0$. To obtain the latter, one can use the ``ball carving''
family of \cite{andoni2006near}, where the balls have radius
$wr=cr/2$. Note however that such a family results in query time that
is at least exponential in $d$, which LSH algorithms are precisely
designed to circumvent.

\section{(More) Deterministic algorithms}

A drawback of data structures described in the previous section  is that they allow ``false negatives'': with a controllable but non-zero probability,  the data structure can report nothing even if the ball $B(q,r)$ is non-empty. Although most of the data structures described in the literature have this property, it is possible to design algorithms with stronger guarantees,  including  deterministic ones. 

The first step in this direction was an observation (already made in~\cite{kushilevitz2000efficient}) that for a finite metric $(X,D)$ supported by  $(c,r)$-ANN data structures, it is possible to construct a data structure that  provides accurate answers to {\em all} queries $q \in X$. This is because one can construct and use $O(\log |X|)$ independent data structures, reducing the probability of failure to $\frac{1}{3 |X|}$. By taking a union bound over all $q \in X$, the constructed data structure works, with probability at least $2/3$, for all queries $X$.  Note that the space and query time  bounds of the new data structure are $O(\log |X|)$ times larger than the respective bounds for 
$(c,r)$-ANN .   Unfortunately, the algorithm for constructing such data structures has still a non-zero failure probability, and no deterministic polynomial-time algorithm for this task is known. 

The first deterministic polynomial-time algorithm for constructing a data structure that works for all queries $q \in X$ appeared in~\cite{indyk2000dimensionality}. It was developed for $d$-dimensional Hamming spaces, and solved a $(c,r)$-ANN  with an approximation factor $c=3+\epsilon$ for any $\epsilon>0$.
The data structure had $d (1/\eps)^{O(1)}$ query time and used $d n^{(1/\eps)^{O(1)}}$ space. 
It relied on two components. 
The first component, ``densification'',  was a deterministic analog of the mapping in Lemma~\ref{lem:KORdr}, which was shown to hold with $k=(d/\epsilon)^{O(1)}$. Retrospectively, the mapping can be viewed as being induced by an adjacency matrix of an {\em unbalanced expander}~\cite{guruswami2009unbalanced}.

\begin{definition} [Expander]
\em An $(r,\alpha)$-{\em unbalanced expander} is a bipartite simple graph $G=(U,V,E)$, $|U|=d, |V|=k$, with
left degree $\Delta$ such that for any $X \subset U$ with $|X| \leq r$, the
set of neighbors $N(X)$ of $X$ has size $|N(X)| \geq (1-\alpha) \Delta |X|$.
\end{definition}

Given such a graph $G$, one can construct a mapping $f=f_G: \{0,1\}^d \to \Sigma^k$ for some finite alphabet $\Sigma$ by letting $f(x)_j$ to be the concatenation of all symbols $x_i$  such that $(i,j) \in E$.
Let $H(x,y)$ be the Hamming distance between $x$ and $y$, i.e., the number of coordinates on which $x$ and $y$ differ. We have that:
\begin{itemize}
\item $H(f(x),f(y))   \le \Delta H(x,y)$, since each difference between $a$ and $b$ contributes to at most $\Delta$ differences between $f(x)$ and $f(y)$. In particular $H(f(x),f(y))   \le \Delta  r (1-\epsilon)$ if $H(x,y) \le  r (1-\epsilon)$.
\item if $H(x,y) \ge r$, then $H(f(x),f(y)) \ge (1-\alpha) \Delta r$ (from the expansion property).
\end{itemize}

Thus, setting $\alpha=\epsilon/2$ yields guarantees analogous to Lemma~\ref{lem:KORdr}, but using a deterministic mapping, and with coordinates of $f(x)$ in $\Sigma$, not $\{0,1\}$. To map into binary vectors, we further replace each symbol $f(x)_j$ by $C(f(x)_j)$, where $C: \Sigma \to \{0,1\}^s$ is an {\em error-correcting code}, i.e., having the property that for any distinct $a,b \in \Sigma$ we have $H(C(a),C(b)) \in [s(1/2-\epsilon), s(1/2+\epsilon)]$. We then use off-the-shelf constructions of expanders~\cite{guruswami2009unbalanced} and codes~\cite{guruswamiessential} to obtain the desired mapping $g=C \circ f: \{0,1\}^d \to  \{0,1\}^{ks}$.

The second component partitions the coordinates of points $g(x)$  into blocks $S_1 \ldots S_t$ of size $\log(n)/\epsilon^{O(1)}$ such that an analog of Lemma \ref{lem:KORdr} holds for all projections  $g(x)_{S_l}$ and $g(y)_{S_l}$ where $x, y \in P$, $l=1 \ldots t$. Such a partitioning can be shown to exist using the probabilistic method, and can be computed deterministically in time polynomial in $n$ via the method of conditional probabilities. Unfortunately, this property does not extend to the case where one of the points (say, $x$) is a query point from $X-P$.  Nevertheless, by averaging, there must be at least one block $S_l$ such that $H(g(x)_{S_l}, g(y)_{S_l} ) \le H(g(x),g(y))/t$, where $y$ is the nearest neighbor of $x$ in $P$. It can be then shown that an approximate near neighbor of $g(x)_{S_l}$ in $\{g(y)_{S_l}: y \in P\}$ is an approximate nearest neighbor of $x$ in $P$. Finding the nearest neighbor in the space restricted to a single block $S_l$  can be solved via exhaustive storage using $n^{1/\epsilon^{O(1)}}$ space, as in Theorem~\ref{t:KOR}.
 
 Perhaps surprisingly, the above construction  is the only known example of a polynomial-size deterministic approximate near neighbor data structure with a constant approximation factor. However, more progress has been shown for an ``intermediary'' problem, where the data structure avoids false negatives by reporting a special symbol $\perp$.
 
 \begin{definition}[$(c,r,\delta)$-Approximate Near Neighbor Without False Negatives (ANNWFN)]
\label{d:crdann}
Given a set $P$ of $n$ points in a metric space $(X,D)$, build a data structure ${\cal S}$ that, given any
query point $q \in X$ such that  $B(q,r) \cap P \neq \emptyset$, ${\cal S}$ returns an element of $(B(q,cr) \cap P) \cup \{\bot\}$, and  $\Pr_{\cal S} [ {\cal S}  \mbox{\ returns\ }  \bot] \le \delta$.
\end{definition}

A  $(1+\epsilon,r,\delta)$-ANNWFN data structure with bounds similar to those in Theorem~\ref{t:KOR} was given in \cite{indyk2000dimensionality}. It used densification and random block partitioning as described above. However, thanks to randomization, block partitioning  could be assumed to hold even for the query point with high probability. 

Obtaining ``no false negatives'' analogs of Theorem~\ref{t:im} turned out to be more difficult. The first such data structure was presented in~\cite{pagh2016locality}, for the Hamming space, achieving query time of the form (roughly) $dn^{1.38/c}$.  Building on that work, very recently, Ahle~\cite{ahle2017optimal} improved the bound to (roughly) $dn^{1/c}$, achieving the optimal runtime exponent.

 In addition to variants of densification and  random block partitioning, the latter algorithm uses a generalization of the space partitioning method from Section~\ref{ss:lsh}, called {\em locality sensitive filtering}.  Such objects can be constructed deterministically in time and space roughly exponential in the dimension. Unfortunately, random block partitioning leads to blocks whose length is larger than $\log n$ by at least a (large) constant, which results in large (although polynomial) time and space bounds. To overcome this difficulty, \cite{ahle2017optimal} shows how to combine filters constructed for dimension $d$ to obtain a filter for dimension $2d$.   This is achieved by using {\em splitters}~\cite{naor1995splitters}, which can be viewed as families of partitions  of $\{1 \ldots 2d\}$ into pairs of sets $(S_1, \overline{S_1}),  (S_2, \overline{S_2}), \ldots$  of size $d$, such that for any $x,y$, there is a pair  $(S_l,\overline{S}_l)$ for which  $H(x_{S_l}, y_{S_l}) = H(x_{\overline{S}_l}, y_{\overline{S}_l}) \pm 1$.   The construction multiplies the space bound by a factor quadratic in $d$, which makes it possible to apply it a small but super-constant number of times to construct filters for (slightly) super-logarithmic dimension. 

\section{Data-dependent approach}
\label{s:dd}

In the earlier sections, we considered ANN data structures that are
based on random and deterministic space partitions. The unifying
feature of all of the above approaches is that the partitions used are
independent of the dataset. This ``data-independence'' leads to
certain barriers: for instance, the \emph{best possible} LSH exponent
is $\rho \ge 1 / c-o(1)$ for the $\ell_1$ distance and $\rho \ge 1 /
c^2 - o(1)$ for $\ell_2$ (see Section~\ref{sec:lshLB}). In this
section, we show how to improve upon the above results significantly
if one allows partitions to depend on the dataset.

This line of study has been developed in a sequence of recent results~\cite{andoni2014beyond,andoni2015optimal,andoni2017optimal}. However,
even before these works, the data-dependent approach had been very popular in practice (see, e.g., surveys~\cite{wang2014hashing,wang2016learning}). Indeed, real-world datasets often have some implicit or explicit structure, thus it pays off to tailor space partitions to a dataset at hand. However, the  theoretical results from~\cite{andoni2014beyond,andoni2015optimal,andoni2017optimal} improve upon data-independent partitions for \emph{arbitrary} datasets. Thus, one must show that \emph{any} set of $n$ points has some structure that makes the ANN problem easier.

\subsection{The result}

In~\cite{andoni2015optimal} (improving upon~\cite{andoni2014beyond}), the following result has been shown.

\begin{theorem}
\label{t:optimal}
  For every $c > 1$, there exists a data structure for $(c, r)$-ANN over $(\Rbb^d, \ell_2)$ with space $n^{1 + \rho} + O(nd)$ and query time $n^{\rho} + dn^{o(1)}$,
  where
  $$
  \rho \leq \frac{1}{2c^2 - 1} + o(1).
  $$
\end{theorem}

This is much better than the best LSH-based data structure, which has
$\rho = \frac{1}{c^2} + o(1)$. For instance, for $c = 2$, the above
theorem improves the query time from $n^{1/4 + o(1)}$ to $n^{1/7 +
  o(1)}$, while using less memory.

Next, we describe the new approach at a high level.

\subsection{Simplification of the problem}

Before describing new techniques, it will be convenient to introduce a few simplifications.
First,
we can assume that $d = \log^{1 + o(1)} n$, by applying Lemma~\ref{lem:JL}. Second, we can reduce the
general ANN problem over $(\Rbb^d, \ell_2)$ to the \emph{spherical} case: where dataset and queries lie on the unit sphere $S^{d-1} \subset \Rbb^d$~(see \cite{razenshteyn2017high}, pages 55--56).
Both the dimension reduction and the reduction to the spherical case incur a negligible loss in the approximation\footnote{Approximation $c$ reduces to approximation $c - o(1)$.}. After the reduction to the
spherical case, the distance to the near neighbor $r$ can be made to be any function of the number of points $n$ that
tends to zero as $n \to \infty$ (for example, $r = \frac{1}{\log \log n}$).

\subsection{Data-independent partitions for a sphere}

In light of the above discussion, we need to solve the $(c, r)$-ANN problem for $S^{d-1}$, where $d = \log^{1 + o(1)} n$ and $r = o(1)$.
Even though the final data structure is based on data-dependent partitions, we start with developing a \emph{data-independent} LSH scheme
for the unit sphere, which will be later used as a building block.

The LSH scheme is parametrized by a number $\eta > 0$. Consider a sequence of i.i.d.\ samples from a standard $d$-dimensional
Gaussian distribution $N(0, 1)^d$: $g_1, g_2, \ldots, g_t, \ldots \in \Rbb^d$. The hash function $h(x)$
of the point $x \in S^{d-1}$ is then defined as $\min_t \{t \geq 1 \mid \langle x, g_t\rangle \geq \eta\}$. This LSH family gives
the following exponent $\rho$ for distances $r$ and $cr$:
\begin{equation}
  \label{rho_ball_carving}
\rho = \frac{\log 1 / p_1}{\log 1 / p_2} = \frac{4 - c^2 r^2}{4 - r^2} \cdot \frac{1}{c^2} + \delta(r, c, \eta),
\end{equation}
where $\delta(r, c, \eta) > 0$ and $\delta(r, c, \eta) \to 0$ as $\eta \to \infty$. Thus, the larger the value of the threshold $\eta$ is, the more efficient the resulting
LSH scheme is. At the same time, $\eta$ affects the efficiency of hash functions.
Indeed, one can show that with very high probability $\max_{x \in S^{d-1}} h(x) \leq e^{(1 + o(1)) \eta^2 / 2} \cdot d^{O(1)}$,
 which bounds the hashing time as well as the number of Gaussian vectors to store.

Consider the expression~(\ref{rho_ball_carving})
for the exponent $\rho$ in more detail. If $r = o(1)$, then we obtain $\rho = \frac{1}{c^2} + o(1)$,
which matches the guarantee of the best data-independent LSH for $\ell_2$.
This is hardly surprising, since, as was mentioned above, the general ANN problem over $\ell_2$ can be reduced to the $(c, r)$-ANN problem over the sphere for $r = o(1)$.
If $r \approx 2/c$, then $\rho$ is close to zero, and, indeed, the $(c, 2 / c)$-ANN problem on the sphere
is trivial (any point can serve as an answer to any valid query).

Between these two extremes, there is a point $r \approx \frac{\sqrt{2}}{c}$ that is crucial for the subsequent discussion.
Since the distance between a pair of random points on $S^{d-1}$ is close to $\sqrt{2}$ with high probability, the problem where $r$ is \emph{slightly} smaller than  $\frac{\sqrt{2}}{c}$
has the following interpretation: if one is guaranteed to have a data point within distance $r$ from the query,
 find a data point that is a bit closer to the query than a typical point on the sphere.
For $r \approx \frac{\sqrt{2}}{c}$, the equation~(\ref{rho_ball_carving}) gives exponent $\rho \approx \frac{1}{2c^2 - 1}$, which
is significantly smaller than the bound $\frac{1}{c^2}$ one is getting for $r = o(1)$. Later, using a certain data-dependent partitioning procedure, we will be able to reduce the \emph{general} ANN problem on the sphere to this intermediate case of $r \approx \frac{\sqrt{2}}{c}$,
thus obtaining the ANN data structure with the exponent $\rho = \frac{1}{2c^2 - 1} + o(1)$. This significantly improves upon the best possible LSH for $\ell_2$ from Section~\ref{s:di}, which yields $\rho = \frac{1}{c^2} + o(1)$.

\subsection{Data-dependent partitions}

We now describe at a high level how to obtain a data structure with space $n^{1 + \rho}$ and query time $n^{\rho}$, where $\rho = \frac{1}{2c^2 - 1} + o(1)$, for the $(c, r)$-ANN problem on the sphere for general $r > 0$. If $r \geq \frac{\sqrt{2}}{c} - o(1)$, then we
can simply use the data-independent LSH described above. Now suppose $r$ is nontrivially smaller than $\frac{\sqrt{2}}{c}$.

We start with finding and removing \emph{dense low-diameter clusters}. More precisely, we repeatedly
find a point $u \in S^{d-1}$ such that $|P \cap B(u, \sqrt{2} - \eps)| \geq \tau n$, where $\eps, \tau = o(1)$, and set
$P := P \setminus B(u, \sqrt{2} - \eps)$. We stop when there are no more dense clusters remaining.
Then we proceed with clusters and the remainder separately.
Each cluster  is enclosed in a ball of radius $1 - \Omega(\eps^2)$ 
and processed recursively.
For the remainder, we sample one partition from the data-independent LSH family described above, apply it to the dataset, and
process each resulting part of the dataset recursively. During the query stage, we (recursively) query the data structure
for \emph{every} cluster (note that the number of clusters is at most $1 / \tau$), and for the remainder we query (again, recursively) a part of the partition, where the query belongs to.
Each step of the aforementioned procedure makes progress as follows. For clusters, we decrease the radius by a factor of $1 - \Omega(\eps^2)$.
It means that we come slightly closer to the ideal case of $r \approx \frac{\sqrt{2}}{c}$, and the instance corresponding to the cluster
becomes easier. For the remainder, we use the fact that there are at most $\tau n$ data points closer than $\sqrt{2} - \eps$ to the query.
Thus, when we apply the data-independent LSH, the expected number of data points in the same part as the query is at most $(\tau + p_2) n$,
where $p_2$ is the probability of collision under the LSH family for points at the distance $\sqrt{2} - \eps$. We set $\tau \ll p_2$,
thus the number of colliding data points is around $p_2 n$. At the same time, the probability of collision with the near neighbor
is at least $p_1$, where $p_1$ corresponds to the distance $r$. Since $r < \frac{\sqrt{2}}{c}$, we obtain an effective exponent
of at most $\frac{1}{2c^2 - 1} + o(1)$. Note that we need to keep extracting the clusters recursively to be able to apply the above reasoning about the remainder set
in each step.

One omission in the above high-level description is that the clusters are contained in smaller \emph{balls} rather than \emph{spheres}. This is
handled by partitioning balls into thin annuli and treating them as spheres (introducing negligible distortion).

\subsection{Time--space trade-off}

In~\cite{andoni2017optimal}, Theorem~\ref{t:optimal} has been extended to provide a smooth time--space trade-off for the ANN problem.
Namely, it allows to decrease the query time at a cost of increasing the space and vice versa.

\begin{theorem}
\label{t:tradeoff}
  For every $c > 1$ and every $\rho_s, \rho_q$ such that
  \begin{equation}
    \label{tradeoff_formula}
  c^2 \sqrt{\rho_q} + (c^2 - 1) \sqrt{\rho_s} \geq \sqrt{2c^2 - 1},
  \end{equation}
  there exists a data structure for $(c, r)$-ANN over $(\Rbb^d, \ell_2)$ with space $n^{1 + \rho_s + o(1)} + O(nd)$ and query time $n^{\rho_q + o(1)} + dn^{o(1)}$.
\end{theorem}

The bound~(\ref{tradeoff_formula}) interpolates between:
\begin{itemize}
\item The near-linear space regime: $\rho_s = 0$, $\rho_q = \frac{2}{c^2} - \frac{1}{c^4}$;
\item The ``balanced'' regime: $\rho_s = \rho_q = \frac{1}{2c^2 - 1}$, where it matches Theorem~\ref{t:optimal};
\item The fast queries regime: $\rho_s = \left(\frac{c^2}{c^2 - 1}\right)^2$, $\rho_q = 0$.
\end{itemize}

For example, for $c = 2$, one can obtain any of the following trade-offs: space $n^{1 + o(1)}$ and query time $n^{7/16 + o(1)}$,
space $n^{8/7 + o(1)}$ and query time $n^{1/7+o(1)}$, and space $n^{16/9+o(1)}$ and query time $n^{o(1)}$.

Theorem~\ref{t:tradeoff} significantly improves upon the previous ANN data structures in various regimes~\cite{indyk1998approximate, kushilevitz2000efficient, indyk2000high, panigrahy2006entropy, kapralov2015smooth}. For example, it improves the dependence on $\eps$ in Theorem~\ref{t:im}
from $O(\log(1 / \eps) / \eps^2)$ to $O(1 / \eps^2)$.

\subsection{Impossibility results}

Similarly to the data-independent case, it is natural to ask whether
exponent $\rho=\tfrac{1}{2c^2-1}+o(1)$ from Theorem~\ref{t:optimal} is
optimal for data-dependent space partitions. In~\cite{andoni16tight},
it was shown that the above $\rho$ is near-optimal in a properly
formalized framework of data-dependent space partitions. This
impossibility result can be seen as an extension of the results
discussed in Section~\ref{sec:lshLB}.

Specifically, \cite{andoni16tight} show that $\rho \geq \frac{1}{2c^2
  - 1}$, where $\rho = \frac{\log 1 / p_1}{\log 1 / p_2}$ for $p_1$
and $p_2$ being certain natural counterparts of the LSH collision
probabilities for the data-dependent case, even when we allow the
distribution on the partitions to depend on a dataset. This result
holds under two further conditions. First, as in
Section~\ref{sec:lshLB}, we need to assume that $p_2$ is not too
small.

The second condition is specific to the data-dependent case, necessary
to address another necessary aspect of the space partition. For any
dataset, where all the points are sufficiently well separated, we can
build an ``ideal'' space partition, with $\rho=0$, simply by
considering its Voronoi diagram. However, this is obviously not a
satisfactory space partition: it is algorithmically hard to compute
fast where in the partition a fixed query point $q$ falls to --- in
fact, it is precisely equivalent to the original nearest neighbor
problem! Hence, to be able to prove a meaningful lower bound on
$\rho$, we would need to restrict the space partitions to have low
run-time complexity (e.g., for a given point $q$, we can compute the
part where $q$ lies in, in time $n^{o(1)}$). This precise restriction
is well beyond reach of the current techniques (it would require
proving computational lower bounds). Instead, \cite{andoni16tight} use
a different, proxy restriction: they require that the
\emph{description complexity} of partitions is $n^{1 -
  \Omega(1)}$. The latter restriction is equivalent to saying that the
distribution of partitions (which may depend on the given dataset) is
supported on a fixed (universal) family of partitions of the size
$2^{n^{1 - \Omega(1)}}$. This restriction, for instance, rules out the
Voronoi diagram, since the latter has a description complexity of
$\Omega(n)$. Furthermore, the description complexity of a randomized
partition is a good proxy for the run-time complexity of a partition
because in all the known constructions of random space partitions with
a near-optimal $\rho$, the run-time complexity is at least the
description complexity, which makes the requirement meaningful.

Overall, under the above two conditions, \cite{andoni16tight} show
that $\rho\ge \tfrac{1}{2c^2-1}-o(1)$ for data-dependent random space
partitions, and hence Theorem~\ref{t:optimal} is essentially optimal in this
framework.

\subsection{ANN for $\ell_\infty$}
\label{sec_linfty}

In this subsection we will describe another type of data-dependent data structure, for the $\ell_{\infty}$ norm. Historically, this was the first example of a data-dependent partitioning procedure used for ANN over high-dimensional spaces.
\begin{theorem}[\cite{indyk2001approximate}]
  For every $0 < \eps < 1$, there exists a \emph{deterministic} data structure for $(c, 1)$-ANN for $(\Rbb^d, \ell_\infty)$ with
approximation $c = O\left(\frac{\log \log d}{\eps}\right)$, space $O(d n^{1 + \eps})$ and query time $O(d \log n)$.
\end{theorem}

The algorithm relies on the following structural lemma.

\begin{lemma}
  Let $P \subset \Rbb^d$ be a set of $n$ points and $0 < \eps < 1$. Then:
  \begin{enumerate}
  \item Either there exists an $\ell_\infty$-ball of radius $O\left(\frac{\log \log d}{\eps}\right)$
    that contains $\Omega(n)$ points from $P$, or
  \item There exists a ``good'' coordinate $i \in \{1, 2, \ldots, d\}$ and a threshold $u \in \Rbb$ such that
    for the sets $A = \{p \in P \mid p_i < u - 1\}$, $B = \{p \in P \mid u - 1 \leq p_i \leq u + 1\}$
    and $C = \{p \in P \mid p_i > u + 1\}$ one has:
    \begin{equation}
      \label{replication_bound}
    \left(\frac{|A| + |B|}{n}\right)^{1 + \eps} + \left(\frac{|B| + |C|}{n}\right)^{1 + \eps} \leq 1
    \end{equation}
    and $|A| / n, |C| / n \geq \Omega(1 / d)$.
  \end{enumerate}
\end{lemma}

Using this lemma, we can build the data structure for $(c, 1)$-ANN for $(\Rbb^d, \ell_\infty)$ recursively.   If there exists a ball $B(x, R)$ with $R = O\left(\frac{\log \log d}{\eps}\right)$
such that $|P \cap B(x, R)| \geq \Omega(n)$ (Case 1), then we store $x$ and $R$ and continue partitioning $P \setminus B(x, R)$ recursively.
If there exists a good coordinate $i \in \{1, 2, \ldots, d\}$ and a threshold $u \in \Rbb$ (Case 2), then we define sets $A$, $B$, $C$
as in the above lemma and partition $A \cup B$ and $B \cup C$ recursively.
We stop as soon as we reach a set that consists of $O(1)$ points.

The query procedure works as follows. Suppose there is a point in $P$ within distance $1$ from $q$ (``the near neighbor''). If we are in Case 1,
 we check if the query point $q$ lies in $B(x, R + 1)$. If it does, we return any data point from $B(x, R)$; f not, we query the
remainder recursively. On the other hand, if we are in Case 2, we query $A \cup B$ if $q_i \leq u$, and $B \cup C$ otherwise. In this case we recurse on the part which is guaranteed to contain a near neighbor.

Overall, we always return a point within distance $O\left(\frac{\log \log d}{\eps}\right)$, and it is straightforward to bound the query time
by bounding the depth of the tree. We obtain the space bound of $O(d n^{1 + \eps})$ by using the property~(\ref{replication_bound}) to bound
the number of times points that are replicated in the Case 2 nodes.

Surprisingly, the approximation $O(\log \log d)$ turns out to be
\emph{optimal} in certain restricted models of
computation~\cite{andoni2008hardness,kapralov2012nns}, including for
the approach from \cite{indyk2001approximate}.

\section{Closest pair}

A problem closely related to ANN is the closest pair problem, which
can be seen as an ``offline'' version of ANN. 
Here, we are given a set
$P$ of $n$ points, and we need to find a pair  $p,q\in P$ of distinct points that
minimize their distance. 

A trivial solution is to compute the distance
between all possible ${n \choose 2}$ pairs of points and take the one
that minimizes the distance. However this procedure has quadratic running time.
 As for the nearest neighbor problem, there is evidence
that for, say, $d$-dimensional $\ell_2$ space, the closest pair
problem cannot be solved in time $n^{2-\alpha}d^{O(1)}$ for any constant $\alpha>0$.

As with $c$-ANN, we focus on the approximate version of the
problem. Furthermore, we consider the {\em decision version}, where we
need to find a pair of points that are below a certain threshold
$r$. The formal definition (for the randomized variant) follows.

\begin{definition}[$(c,r)$-approximate close pair problem, or $(c,r)$-CP]
Given a set of points $P\subset X$ of size $n$, if there exist
distinct points $p^*,q^*\in X$ with $\d(p^*, q^*)\le r$,  find a
pair of distinct points $p,q \in P$ such that $\d(p,q)\le cr$, with probability at
least 2/3.
\end{definition}

The $(c,r)$-CP problem is closely related to the $(c,r)$-ANN problem
because we can solve the former using a data structure for the
latter. In particular, one can run the following procedure: 
partition $P$ into two sets $A$, $B$ randomly; build $(c,r)$-ANN
on the set $A$;  query every point $q\in B$. It is easy to see that
one such run succeeds in solving a $(c,r)$-approximate close pair with
probability at least $1/2\cdot 2/3$. Repeating the procedure 3 times
is enough to guarantee a success probability of $2/3$. If $(c,r)$-ANN
under the desired metric can be solved with query time $Q(n)$ and
preprocessing time $S(n)$, we obtain a solution for $(c,r)$-CP running
in time $O(S(n)+nQ(n))$. For example, applying the reduction from above for $(\R^d, \ell_p)$
space for $p\in\{1,2\}$, we immediately obtain an algorithm running in
$O(dn^{1+\rho})$ time, where $\rho=\tfrac{1}{2c^p-1}+o(1)$ (Section 
\ref{s:dd}).

Focusing on the case of $\ell_2$, and approximation $c=1+\eps$, the
above algorithm has runtime $O(n^{2-4\eps+O(\eps^2)}d)$. It turns out
that, for the $\ell_2$ norm,  one can obtain algorithms with a better dependance on $\eps$, for
small $\eps$. In particular, the line of work from
\cite{valiant2015finding, karppa2016faster, alman2016polynomial} led to
the following algorithm:

\begin{theorem}[\cite{alman2016polynomial}]
\label{thm:closestPair}
Fix dimension $d\ge 1$, $r>0$, and $\eps\in(0,1/2)$. Then, for any set of $n$
points in $\R^d$, one can solve the $(1+\eps,r)$-CP over $\ell_2$ in time
$O(n^{2-\Omega(\eps^{1/3}/\log(1/\eps))}+dn)$, with constant probability.
\end{theorem}
Note that the running time bound in the above theorem is better than that obtained using  LSH data
structures, for small enough $\eps$.

The main new technical ingredient 
is the
     {\em fast matrix multiplication} algorithm. In particular,
     suppose we want to multiply two matrices of size $n\times m$ and
     $m\times n$. Doing so na\"ively takes time $O(n^2m)$. Starting
     with the work of \cite{strassen1969gaussian}, there has been
     substantial work to improve this run-time; see also
     \cite{williams2012multiplying}. Below we state the running time of a fast matrix multiplication algorithm
     due to \cite{coppersmith1982rapid}, which is most relevant for
     this section.

\begin{theorem}[\cite{coppersmith1982rapid}]
\label{thm:fmm}
Fix $n\ge 1$ and let $m\ge1$ be such that $m\le n^{0.172}$. One can
compute the product of two matrices of sizes $n\times m$ and $m\times
n$ in $O(n^2\log^2n)$ time.
\end{theorem}

\subsection{Closest pair via matrix multiplication}

We now sketch the algorithm for the closest pair from
\cite{valiant2015finding}, which obtains
$O(n^{2-\Omega(\sqrt{\eps})}d)$ time. The algorithm is best described
in terms of {\em inner products}, as opposed to distances as
before. In particular, suppose we have a set of points $P\subset
{\mathcal S}^d$ of unit norm, where all pairs of points have inner
product in the range $[-\theta,\theta]$, except for one ``special''
pair that has inner product at least $c\theta$, for some scale
$\theta>0$ and approximation $c=1+\eps$. Now
the problem is to find this special pair---we term this problem
$(c,\theta)$-IP problem. We note that we can reduce 
$(1+\eps,r)$-CP over $\ell_2$ to  $(1+\Omega(\eps),1/2)$-IP
, by using the embedding of \cite{schoenberg1942positive}, or
Lemma~\ref{lem:KORdr} of \cite{kushilevitz2000efficient}.

A natural approach to the the IP problem is to multiply two $n\times
d$ matrices: if we consider the matrix $M$ where the rows are the
points of $P$, then $MM^t$ will have a large off-diagonal entry
precisely for the special pair of points. This approach however
requires at least $n^2$ computation time, since even the output of
$MM^t$ has size $n^2$. Nevertheless, an extension of this approach
gives a better run-time when $c$ is very large (and hence $\theta<1/c$
very small, i.e., all points except for the special pair are
near-orthogonal). In particular, partition randomly the vectors from
$P$ into $n/g$ groups $S_1,\ldots S_{n/g}$, each of size $O(g)$.
For each group $i$, we sum the vectors $S_i$ with random signs,
obtaining vectors $v_i=\sum_{p_j\in S_i} \chi_j p_j$, where $p_j$ are
the points in $P$ and $\chi_j$ are Rademacher random variables.  Now
the algorithm forms a matrix $M$ with $v_i$'s as rows, and computes
$MM^t$ using fast matrix multiplication (Theorem~\ref{thm:fmm}).  The
two special points are separated with probability $1-g/n$.
Conditioning on this event, without loss of generality, we can assume
that they are in group 1 and 2 respectively.  Then, it is easy to note
that $|(MM^t)_{12}|\approx \Theta(c\cdot \theta)$, whereas, for
$(i,j)\neq (1,2)$ and $i\neq j$, we have that $|(MM^t)_{ij}|\approx
O(g\cdot \theta)$ with constant probability. Hence, we can identify
the special pair in the product $MM^t$ as long as $c\gg g$, and yields
runtime $O(n^2/g^2)$, i.e., a $g^2\ll c^2$ speed-up over the na\"ive
algorithm (note that Theorem~\ref{thm:fmm} requires that $d<n^{0.172}$).

The above approach requires $c$ to be very large, and hence the
challenge is whether we can reduce the case of $c=1+\eps$ to the case
of large $c$. Indeed, one method is to use tensoring: for a fixed
parameter $k$ and any two vectors $x,y\in\R^d$, we consider
$x^{\otimes k}, y^{\otimes k}\in \R^{d^k}$, for which $\langle
x^{\otimes k},y^{\otimes k}\rangle=(\langle x,y\rangle)^k$. Thus
tensoring reduces the problem of $(1+\eps, 1/2)$-IP to
$((1+\eps)^k,2^{-k})$-IP, and hence we hope to use the above algorithm
for $c=(1+\eps)^k\approx e^{\eps k}$. If we use $t=\zeta\ln n$, for
small constant $\zeta$, we obtain $c=n^{\eps\zeta}$, and hence we
obtain a speed-up of $g^2\approx c^{2}=n^{2\eps\zeta}$. One caveat here is
that, after tensoring the vectors, we obtain vectors of dimension
$d^k$, which could be much larger than $n$---then even writing down
such vectors would take $\Omega(n^2)$ time. Yet, one can use a
dimension reduction method, like Lemma~\ref{lem:JL}, to reduce
dimension to $O(\tfrac{\log n}{\theta^k})=\tilde O(n^{\zeta \ln 2})$,
which is enough to preserve all inner products up to additive, say,
$0.1\cdot\theta^k$. There are further details (e.g., we cannot afford to
get high-dimensional vectors in the first place, even if we perform
dimension-reduction), see \cite{valiant2015finding, karppa2016faster}
for more details.

The above algorithm yields a speed-up of the order of $n^{O(\eps)}$,
i.e., comparable to the speed-up via the LSH methods. To obtain a
better speed-up, like in the Theorem~\ref{thm:closestPair}, one can
replace the tensoring transformation with a more efficient
one. Indeed, one can employ an {\em asymmetric} embedding $f,g:\R^d\to
\R^{m}$, with the property that for any unit-norm vectors $x,y$, we
have that $\langle f(x),g(y)\rangle=p(\langle x,y\rangle)$, where
$p(\cdot)$ is a polynomial of choice. In particular, we require a
polynomial $p(\cdot)$ that is small on the interval $[-\theta,
  \theta]$, as large as possible on $[(1+\eps)\theta, 1]$, and $p(1)$
is bounded. Note that the tensoring operation implements such an
embedding with $p(a)=a^k$ and where $f(x)=g(x)=x^{\otimes
  k}$. However, there are more efficient polynomials: in fact, the
optimal such polynomial is the Chebyshev polynomial. For example, for
the degree-$k$ Chebyshev polynomial $T_k(\cdot)$, we have that
$T_k(1+\eps)/T_k(1)\approx e^{\sqrt{\eps}k}$, which is in contrast to the
above polynomial $p(a)=a^k$, for which $p(1+\eps)/p(1)\approx e^{\eps
  k}$.

Using the Chebyshev polynomials, one can obtain a runtime of
$n^{2-\Omega(\sqrt{\eps})}$ for the IP and hence CP problem. To obtain
the improved result from Theorem~\ref{thm:closestPair},
\cite{alman2016polynomial} employ {\em randomized} polynomials, i.e.,
a distribution over polynomials where $p(\cdot)$ is small/large only
with a certain probability. Without going into further details, the
theorem below states the existence of such polynomials, which are used
to obtain $n^{2-\Omega(\eps^{1/3} /\log(1/\eps))}$ run-time for the
$(1+\eps, r)$-CP problem.

\begin{theorem}[\cite{alman2016polynomial}]
Fix $d\ge 1$, $\theta\ge1$, $s\ge1$, and $\eps>0$. There exists a
distribution over polynomials $P:\{0,1\}^d\to \R$ of degree
$O(\eps^{-1/3}\log s)$, such that we have the following for any $x\in
\{0,1\}^d$:
\begin{itemize}
\item
if $\sum_{i=1}^d x_i\le \theta$, then $|P(x)|\le 1$ with probability at
least $1-1/s$;
\item
if $\sum_{i=1}^d x_i\in (\theta,(1+\eps)\theta)$, then $|P(x)|>1$ with probability at
least $1-1/s$;
\item
if $\sum_{i=1}^d x_i>(1+\eps)\theta$, then $|P(x)|\ge s$ with probability at
least $1-1/s$.
\end{itemize}
\end{theorem}

\section{Extensions}
\label{s:extensions}

In this section, we discuss several techniques that significantly extend the class of spaces which admit efficient ANN data structures.

\subsection{Metric embeddings}

So far, we have studied the ANN problem over the $\ell_1$, $\ell_2$ and $\ell_\infty$ distances.
A useful approach is to \emph{embed} a metric of interest into $\ell_1/\ell_2/\ell_\infty$
and use one of the data structures developed for the latter spaces.

\subsubsection{Deterministic embeddings}

\begin{definition}
\label{def_bilip}
  For metric spaces $\mathcal{M} = (X, \d_X)$, $\mathcal{N} = (Y, \d_Y)$ and for $D \geq 1$, we say that a map $f \colon X \to Y$
  is a \emph{bi-Lipschitz embedding} with distortion $D$ if there exists $\lambda > 0$ such that for every
  $x_1, x_2 \in X$ one has:
  $$
  \lambda d_X(x_1, x_2) \leq \d_Y(f(x_1), f(x_2)) \leq D \cdot \lambda \d_X(x_1, x_2).
  $$
\end{definition}

A bi-Lipschitz embedding of $\mathcal{M}$ into $\mathcal{N}$ with distortion $D$
together with a data structure for $(c, r)$-ANN
over $\mathcal{N}$ immediately implies a data structure for $(cD, r')$-ANN over $\mathcal{M}$,
where $r' = \frac{r}{\lambda D}$.
However, space and query time
of the resulting data structure depend crucially on the \emph{computational efficiency} of the embedding, since,
in particular, the query procedure requires evaluating the embedding on a query point.

As the following classic results show, any finite-dimensional normed or finite metric space can be embedded into finite-dimensional
$\ell_\infty$ with small distortion.

\begin{theorem}[Fr\'{e}chet--Kuratowski, \cite{frechet1906quelques,kuratowski1935quelques}]
  \label{metric_linfty}
  If $\mathcal{M}$ is a finite metric space, which consists of $N$ points, then
  $\mathcal{M}$ embeds into $(\Rbb^N, \ell_\infty)$ with distortion $D = 1$ (isometrically).
\end{theorem}

\begin{theorem}[see, e.g., \cite{wojtaszczyk1996banach}]
  \label{norm_linfty}
  For every $\eps > 0$, every normed space $(\Rbb^d, \|\cdot\|)$
  embeds with distortion $1 + \eps$ into $(\Rbb^{d'}, \ell_\infty)$,
  where $d' = O(1 / \eps)^d$, via a linear map.
\end{theorem}

However, the utility of Theorems~\ref{metric_linfty} and~\ref{norm_linfty} in the context of the ANN problem is limited,
since the required dimension of the target $\ell_\infty$ space is very high (in particular, Theorem~\ref{norm_linfty} gives a data
structure with \emph{exponential} dependence on the dimension). Moreover, even if we allow the distortion
$D$ of an embedding to be a large constant, the target dimension can not be improved much.
As has been shown in~\cite{matousek1997embedding}, one needs at least $N^{\Omega(1 / D)}$-dimensional $\ell_\infty$ to ``host''
all the $N$-point metrics with distortion $D$. For $d$-dimensional norms, even as simple as $\ell_2$,
the required dimension is $2^{\Omega_D(d)}$~\cite{figiel1977dimension,ball1997elementary}.

More generally, (lower-dimensional) $\ell_\infty$ turns out to be not so useful of a target space, and only a handful
of efficient embeddings into $\ell_\infty$ are known (for instance, such an embedding has been constructed
in~\cite{farach1999approximate} for the Hausdorff distance).
Luckily, the situation drastically improves, if we allow \emph{randomized} embeddings, see Section~\ref{random_embeddings}
for the examples.

Instead of $\ell_\infty$, one can try to embed a metric of interest into $\ell_1$ or $\ell_2$. Let us list a few cases,
where such embeddings lead to efficient ANN data structures.
\begin{itemize}
\item Using the result from~\cite{johnson1982embedding}, one can embed $(\Rbb^d, \ell_p)$ for $1 < p \leq 2$ into $(\Rbb^{d'}, \ell_1)$ with distortion
  $1 + \eps$, where $d' = O(d / \eps^2)$. Moreover, the corresponding map is linear and hence efficient to store and apply.
  This reduction shows that the ANN problem over $\ell_p$ for $1 < p \leq 2$ is no harder than for the $\ell_1$ case.
  However, later in this section we will show how to get a better ANN algorithm for the $\ell_p$ case
  using a different embedding.
\item For the Wasserstein-$1$ distance (a.k.a.\ the Earth-Mover distance in the computer science literature) between probability measures
  defined on $\{1, 2, \ldots, d\}^k$, one can use the results from~\cite{charikar2002similarity,indyk2003fast,naor2007planar}, to embed it into $(\Rbb^{d^{O(k)}}, \ell_1)$ with distortion $O(k \log d)$.
\item The Levenshtein distance (a.k.a. edit distance) over the binary
  strings $\{0, 1\}^d$
  can be embedded into $(\Rbb^{d^{O(1)}}, \ell_1)$
  with distortion $2^{O(\sqrt{\log d \log \log d})}$~\cite{ostrovsky2007low}.
\end{itemize}

Let us note that there exist generic results concerned with embeddings into $\ell_1/\ell_2$ similar to Theorem~\ref{metric_linfty}
and Theorem~\ref{norm_linfty}.

\begin{theorem}[\cite{bourgain1985lipschitz,linial1995geometry}]
  \label{metric_l2}
  Any $N$-point metric embeds into $(\Rbb^{O(\log N)}, \ell_2)$ with distortion $O(\log N)$.
\end{theorem}

\begin{theorem}[\cite{john1948extremum, ball1997elementary}]
  \label{norm_l2}
  Any normed space $(\Rbb^d, \|\cdot\|)$ embeds into $(\Rbb^d, \ell_2)$ with distortion $\sqrt{d}$
  via a linear map.
\end{theorem}

Theorem~\ref{metric_l2} does not give an embedding efficient enough for the ANN applications: computing it in one point
requires time $\Omega(N)$. At the same time, Theorem~\ref{norm_l2} \emph{is} efficient and, together with
$\ell_2$ data structures, gives an ANN data structure for a general $d$-dimensional norm with approximation $O(\sqrt{d})$.

Since the ANN problem is defined for two specific distance scales ($r$ and $cr$), we do not need the full power
of bi-Lipschitz embeddings and sometimes can get away with weaker notions of embeddability.
For example, the following theorem follows from the results of~\cite{schoenberg1937certain}.

In the theorem, $\ell_2(\mathbb{N})$ denotes the space of infinite sequences $(a_i)_{i=1}^\infty$ such that $\sum_i |a_i|^2 < +\infty$
and the norm of the sequence $\|a\|_2$ is equal to $\left(\sum_i |a_i|^2 \right)^{1/2}$.

\begin{theorem}
  For every $1 \leq p < 2$ and every $d \geq 1$, there exists a map $f \colon \Rbb^d \to \ell_2(\mathbb{N})$
  such that for every $x, y \in \Rbb^d$, one has:
  $$
  \|f(x) - f(y)\|_2^2 = \|x - y\|_p^p.
  $$
\end{theorem}

This embedding allows to use an ANN data structure for $\ell_2$ with approximation $c$ to
get an ANN data structure for $\ell_p$ with approximation $c^{2 / p}$. However, for this
we need to make the embedding computationally efficient. In particular, the target
must be finite-dimensional. This can be done, see~\cite{Nguyen-thesis} for details. As a result,
for the $\ell_p$ distance for $1 \leq p < 2$, we are able to get the result
similar to the one given by Theorem~\ref{t:tradeoff}, where in~(\ref{tradeoff_formula}) $c^2$ is replaced with $c^p$ everywhere.

\subsubsection{Randomized embeddings}
\label{random_embeddings}

It would be highly desirable to utilize the fact that every metric embeds well into $\ell_\infty$ (Theorems~\ref{metric_linfty} and~\ref{norm_linfty})
together with the ANN data structure for $\ell_\infty$ from Section~\ref{sec_linfty}.
However, as discussed above, spaces as simple as $(\Rbb^d, \ell_1)$ or $(\Rbb^d, \ell_2)$
require the target $\ell_\infty$ to have $2^{\Omega(d)}$ dimensions
to be embedded with small distortion. It turns out, this can be remedied by allowing embeddings to be \emph{randomized}.
In what follows, we will consider the case of $(\Rbb^d, \ell_1)$, and then generalize the construction to other metrics.

The randomized embedding of $(\Rbb^d, \ell_1)$ into $(\Rbb^d, \ell_\infty)$ is defined as follows: we generate $d$ i.i.d.\ samples $u_1$, $u_2$, \ldots, $u_d$
from the exponential distribution with parameter $1$, and then the embedding $f$ maps
a vector $x \in \Rbb^d$ into
$$
\left(\frac{x_1}{u_1}, \frac{x_2}{u_2}, \ldots, \frac{x_d}{u_d}\right).
$$
Thus, the resulting embedding is linear. Besides that, it is extremely efficient to store ($d$ numbers) and apply ($O(d)$ time).

Let us now understand how $\|f(x)\|_\infty$ is related to $\|x\|_1$. The analysis uses (implicitly) the \emph{min-stability}
property of the exponential distribution. One has for every $t > 0$:

\[  \mathrm{Pr}_f[\|f(x)\|_\infty \leq t] = \prod_{i=1}^d \mathrm{Pr}\left[\frac{|x_i|}{u_i} \leq t\right] 
  = \prod_{i=1}^d \mathrm{Pr}\left[u_i \geq \frac{|x_i|}{t}\right]
  = \prod_{i=1}^d e^{-|x_i| / t}
  = e^{-\|x\|_1 / t}.
\]

The random variable $\|f(x)\|_\infty$ does not have a finite first moment, however its mode is in the point $t = \|x\|_1$,
which allows us to use $\|f(x)\|_\infty$ to estimate $\|x\|_1$.
It is immediate to show that for every $\delta > 0$, there exist $C_1, C_2 > 1$ with $C_1 = O(\log(1 / \delta))$
and $C_2 = O(1 / \delta)$
such that for every $x$, one has:
\begin{equation}
  \label{exp_lower_tail}
\mathrm{Pr}_f\left[\|f(x)\|_\infty \geq \frac{\|x\|_1}{C_1}\right] \geq 1 - \delta
\end{equation}
and
\begin{equation}
  \label{exp_upper_tail}
\mathrm{Pr}_f\left[\|f(x)\|_\infty \leq C_2 \cdot \|x\|_1\right] \geq 1 - \delta
\end{equation}

Thus, the map $f$ has distortion $O\left(\frac{\log(1 / \delta)}{\delta}\right)$ with probability $1 - \delta$.
However, unlike the deterministic case, the randomized guarantees~(\ref{exp_lower_tail}) and~(\ref{exp_upper_tail}) are not sufficient for the reduction between
ANN data structures (if $\delta \gg 1 / n$).
This is because the lower bound on $\|f(x)\|_\infty$ must apply simultaneously to all ``far'' points.
In order to obtain a desired reduction, we need to use slightly different parameters. Specifically,
for $0 < \eps < 1$ one has:
$$
\mathrm{Pr}_f\left[\|f(x)\|_\infty \geq \Omega\left(\frac{\|x\|_1}{\log n}\right)\right] \geq 1 - \frac{1}{10 n}
$$
and
$$
\mathrm{Pr}_f\left[\|f(x)\|_\infty \leq O\left(\frac{\|x\|_1}{\eps \cdot \log n}\right)\right] \geq n^{-\eps}.
$$
This allows us to reduce the $(c / \eps, r)$-ANN problem over $(\Rbb^d, \ell_1)$
to $n^{O(\eps)}$ instances of the $(c, r')$-ANN problem over $(\Rbb^d, \ell_\infty)$.
Indeed, we sample $n^{O(\eps)}$ i.i.d.\ maps $f_i$ as described above and solve the ANN problem over $\ell_\infty$
on the image of $f_i$. Far points remain being far with probability $1 - 1 / 10 n$ each.
Using the linearity of expectation and the Markov inequality,
we observe that, with probability at least $0.9$, \emph{no} far point come close enough to the query point.
At the same time, with probability at least $n^{-\eps}$, the near neighbor does not move too far away, so, with high probability,
at least one of the $n^{O(\eps)}$ data structures succeeds. This reduction is quite similar to the use of Locality-Sensitive Hashing in Section~\ref{ss:lsh}.

As a result, we get an ANN data structure for $(\Rbb^d, \ell_1)$ with approximation $O\left(\frac{\log \log d}{\eps^2}\right)$, query time
$O(d n^{\eps})$ and space $O(dn^{1 + \eps})$. This is worse than the best ANN data structure for $\ell_1$ based on
(data-dependent) space partitions. However, the technique we used is very versatile and generalizes easily to many other distances.
The $\ell_1$ embedding was first used in~\cite{andoni2009overcoming}. Later, it was generalized~\cite{andoni2009nearest} to $\ell_p$ spaces for $p \geq 1$.
To get such an embedding, one can divide every coordinate by the $(1 / p)$-th power of an exponential random variable.
Finally, in~\cite{andoni2017approximate} the same technique has been
shown to work for Orlicz norms and top-$k$ norms, which we define next.

\begin{definition}
  Let $\psi \colon [0; +\infty) \to [0; +\infty)$ be a non-negative monotone increasing convex function with $\psi(0) = 0$. Then,
      an \emph{Orlicz norm} $\|\cdot\|_{\psi}$ over $\Rbb^d$ is given by its unit ball $K_{\psi}$, defined as follows:
      $$
      K_{\psi} = \left\{x \in \Rbb^d \,\middle|\, \sum_{i=1}^d \psi(|x_i|) \leq 1\right\}.
      $$
\end{definition}

Clearly, $\ell_p$ norm for $p < \infty$ is Orlicz for $\psi(t) = t^p$.

\begin{definition}
  For $1 \leq k \leq d$, we define the \emph{top-$k$ norm} of a vector from $\Rbb^d$ as the sum of $k$ largest absolute values of the
  coordinates.
\end{definition}

The top-$1$ norm is simply $\ell_\infty$, while top-$d$ corresponds to $\ell_1$.

To embed an Orlicz norm $\|\cdot\|_\psi$ into $\ell_\infty$, we divide the coordinates using a random variable $X$
with the c.d.f.\ $F_X(t) = \mathrm{Pr}[X \leq t] = 1 - e^{-\psi(t)}$. To embed the top-$k$ norm, we use a truncated exponential distribution.
All of the above embeddings introduce only a constant distortion.

Let us note that for the $\ell_p$ norms one can achieve approximation $2^{O(p)}$~\cite{naor2006approximate,bartal2015approximate}, which is an improvement upon the above $O(\log \log d)$
bound if $p$ is sufficiently small.

\subsection{ANN for direct sums}
\label{direct_sums_ann}

In this section we describe a vast generalization of the ANN data structure for $\ell_\infty$ from Section~\ref{sec_linfty}.
Namely, we will be able to handle \emph{direct sums} of metric spaces.

\begin{definition}
  Let $M_1 = (X_1, \d_1)$, $M_2 = (X_2, \d_2)$, \ldots, $M_k = (X_k, \d_k)$ be metric spaces and let
  $\|\cdot\|$ be a norm over $\Rbb^k$.
  Then the \emph{$\|\cdot\|$-direct sum} of $M_1$, $M_2$, \ldots, $M_k$ denoted by $\left(\bigoplus_{i=1}^k M_i\right)_{\|\cdot\|}$ is a metric space defined as follows. The ground set is the Cartesian product $X_1 \times X_2 \times \ldots \times X_k$.
  The distance function $\d$ is given by the following formula.
  $$
  \d\left((x_1, x_2, \ldots, x_k), (y_1, y_2, \ldots, y_k)\right) = \left\|\left(\d_1(x_1, y_1), \d_2(x_2, y_2), \ldots, \d_k(x_k, y_k)\right)\right\|.
  $$
\end{definition}

It turns out that for many interesting norms $\|\cdot\|$ the following holds. If for metrics $M_1$, $M_2$, \ldots, $M_k$
there exist efficient ANN data structures, then the same holds for $\left(\bigoplus_{i=1}^k M_i\right)_{\|\cdot\|}$ (with a mild loss in the parameters).

The first result of this kind was shown in~\cite{indyk2002approximate}\footnote{In~\cite{indyk2002approximate}, a slightly weaker version of Theorem~\ref{linfty_product} has been stated.
First, it assumed \emph{deterministic} data structures for the spaces $M_i$. This is straightforward to address
by boosting the probability of success for data structures for $M_i$ using repetition.
Second, the resulting space bound~\cite{indyk2002approximate} was worse. An improvement
to the space bound has been described in Appendix~A of the arXiv version of~\cite{andoni2017approximate}.
Finally, the paper~\cite{indyk2002approximate} assumes ANN data structures for $M_i$ with a slightly stronger guarantee.
Namely, each point is assigned a \emph{priority} from $1$ to $n$, and if the near neighbor has priority $t$,
we must return a point with priority at most $t$. It is not hard to solve the version with priorities
using a standard ANN data structure (with $\log^{O(1)} n$ loss in space and query time).
A na\"{\i}ve reduction builds an ANN data structure for points with priority at most $t$ for every $t$.
Then, we can run a binary search over the resulting priority. However, this gives a \emph{linear in $n$} loss
in space. To rectify this, we use a standard data structure technique: the decomposition of an interval into $O(\log n)$
\emph{dyadic intervals}, i.e., intervals of the form $[2^k \cdot l + 1; 2^k \cdot (l + 1)]$ for integer $k, l$.. Thus, we build an ANN data structure for every dyadic interval of priorities.
This still gives $O(n)$ ANN data structures, however, each data point participates in at most $O(\log n)$ of them.
}
 for the case of $\ell_\infty$-direct sums.
In what follows we denote by $d$ the ``complexity'' of each metric $M_i$.
That is, we assume  it takes $O(d)$ time to compute the distance
between two points, and that a point requires $O(d)$ space to store.

\begin{theorem}
  \label{linfty_product}
  Let $c > 1$, $r > 0$ and $0 < \eps < 1$.
  Suppose that each $M_i$ admits a $(c, r)$-ANN data structure for $n$-point sets with space $n^{1 + \rho}$ (in addition to storing the dataset)
  for some $\rho \geq 0$ and
  query time $Q(n)$. Then, there exists a data structure
  for $(c', r)$-ANN over $\left(\bigoplus_{i=1}^k M_i\right)_{\infty}$, where
$c' = O\left(\frac{c \log \log n}{\eps}\right)$, 
the space is $O(n^{1 + \rho + \eps})$ (in addition to storing the dataset), and the
query time is $Q(n) \cdot \log^{O(1)} n + O(dk \log n)$.
\end{theorem}

Informally speaking, compared to data structures for $M_i$, the data structure for $(\bigoplus_i M_i)_{\infty}$
loses $\frac{\log \log n}{\eps}$ in approximation, $n^{\eps}$ in space, and $\log^{O(1)} n$ in query time.

Later, the result of~\cite{indyk2002approximate} was significantly extended~\cite{indyk2004approximate, andoni2009overcoming, andoni2009nearest, andoni2017approximate}, to support $\|\cdot\|$-direct sums where $\|\cdot\|$ is an $\ell_p$ norm, an Orlicz norm, or a top-$k$ norm.
The main insight is that we can use the randomized embeddings of various norms into $\ell_\infty$
developed in Section~\ref{random_embeddings}, to reduce the case of $\|\cdot\|$-direct sums to the case
of $\ell_\infty$-direct sums. Indeed, we described how to reduce
the ANN problem over several classes of norms to $n^{\eps}$ instances of
ANN over the $\ell_\infty$ distance at a cost of losing $O(1 / \eps)$ in the approximation. It is not hard to see that
the exact same approach can be used to reduce the ANN problem over $\left(\bigoplus_{i=1}^k M_i\right)_{\|\cdot\|}$
to $n^{\eps}$ instances of ANN over $\left(\bigoplus_{i=1}^k M_i\right)_{\infty}$ also at a cost of losing $O(1 / \eps)$
in approximation.

\subsection{Embeddings into direct sums}

As Section~\ref{direct_sums_ann} shows, for a large class of norms $\|\cdot\|$, we can get an efficient ANN data structure for any $\|\cdot\|$-direct sum
of metrics that admit efficient ANN data structures. This gives a natural approach to the ANN problem: embed a metric of interest into such a direct sum.

This approach has been successful in several settings.
In~\cite{indyk2002approximate}, the Fr\'{e}chet distance between two sequences of points in a metric space
is embedded into an $\ell_\infty$-direct sums of Fr\'{e}chet distances between shorter sequences.
Together with Theorem~\ref{linfty_product}, this was used to obtain an ANN data structure for the Fr\'{e}chet distance.
In~\cite{andoni2009overcoming}, it is shown how to embed the Ulam
metric (which is the edit distance between permutations of
length $d$)
into $\left(\bigoplus^{d} \left(\bigoplus^{O(\log d)} (\Rbb^d, \ell_1)\right)_{\ell_\infty}\right)_{\ell_2^2}$ with a constant distortion which gives an ANN data structure with doubly-logarithmic approximation.
At the same time, the Ulam distance requires distortion $\Omega\left(\frac{\log d}{\log \log d}\right)$
to embed into $\ell_1$~\cite{andoni2010computational}. This shows that (lower-dimensional) direct sums form a strictly more ``powerful'' class of spaces than $\ell_1$ or $\ell_2$. Finally, in~\cite{andoni2017approximate}, it is shown that \emph{any} symmetric norm over $\Rbb^d$
is embeddable into $\left(\bigoplus_{i=1}^{d^{O(1)}} \left(\bigoplus_{j=1}^d X_{ij}\right)_1\right)_\infty$
with constant distortion, where $X_{ij}$ is $\Rbb^d$ equipped with the top-$j$ norm.
Together with the results from Section~\ref{random_embeddings} and Section~\ref{direct_sums_ann}, this gives an ANN algorithm with approximation $(\log \log n)^{O(1)}$ for general \emph{symmetric}\footnote{Under permutations and negations
of the coordinates.} norms.

\subsection{ANN for general norms}

For \emph{general} $d$-dimensional norms, the best known ANN data
structure is obtained by combining Theorem~\ref{norm_l2}
with an efficient ANN data structure for $\ell_2$ (for example, the
one given by Theorem~\ref{t:optimal}). This approach gives
approximation $O(\sqrt{d / \eps})$ for space $d^{O(1)} \cdot n^{1 +
  \eps}$ and query time $d^{O(1)} \cdot n^{\eps}$ for every constant
$0 < \eps < 1$.  Very recently, the approximation $O(\sqrt{d / \eps})$
has been improved to $O\left(\frac{\log
  d}{\eps^2}\right)$~\cite{andoni2017data} for the same space and time
bounds if one is willing to relax the model of computation to the
\emph{cell-probe model}, where the query procedure is charged for
\emph{memory accesses}, but any computation is free.
This ANN data structure
heavily builds on a recent geometric result
from~\cite{naor2017spectral}: a bi-Lipschitz embedding (see
Definition~\ref{def_bilip}) of the shortest-path metric of \emph{any}
$N$-node expander graph~\cite{hoory2006expander} into an
\emph{arbitrary} $d$-dimensional normed space must have distortion at
least $\Omega\left(\log_d N\right)$.  At a very high level, this
non-embeddability result is used to claim that any large
bounded-degree graph, which \emph{does} embed into a normed space, can
not be an expander, and hence it must have a sparse cut. The existence
of the sparse cut is then used, via a duality argument, to build a
(data-dependent) random space partition family for a general
$d$-dimensional normed space. The latter family is used to
obtain the final data structure.

This approach can be further extended for several norms of interest to
obtain proper, \emph{time-efficient} ANN data structures, with even
better approximations. For instance, \cite{andoni2017data} show how to
get ANN with approximation $O(p)$ for the $\ell_p$ norms, improving
upon the bound $2^{O(p)}$
from~\cite{naor2006approximate,bartal2015approximate}.  Finally, for
the Schatten-$p$ norms of matrices, defined as the $\ell_p$ norm of
the vector of singular values, one obtains approximation $O(p)$ as
well, while the previous best approximation was polynomial in the
matrix size (by relating the Schatten-$p$ norm to the Frobenius norm).

\paragraph{Acknowledgements} The authors would like to thank Assaf Naor, Tal Wagner,  Erik Waingarten  and Fan Wei for many helpful comments. This research was supported by NSF and Simons Foundation. 
\bibliographystyle{alpha}
\bibliography{bibfile}
\end{document}